\documentclass[a4paper,12pt]{article}

\usepackage[truedimen,margin=25truemm]{geometry}
\usepackage[whole]{bxcjkjatype}

\usepackage{comment}
\usepackage{url}
\usepackage{hyperref}

\usepackage{amsmath,amssymb}
\usepackage{exscale}
\usepackage{bm}
\usepackage[hiresbb]{graphicx}
\usepackage{tikz}
\usepackage{amsthm}
\usepackage[hang,small,bf]{caption}
\usepackage[subrefformat=parens]{subcaption}
\usetikzlibrary {arrows.meta}
\usetikzlibrary{calc}
\usepackage{here}
\usepackage{ascmac}
\usetikzlibrary{intersections,calc,arrows}
\newcommand{\ctext}[1]{\raise0.2ex\hbox{\textcircled{\scriptsize{#1}}}}

    \theoremstyle{definition}
    \newtheorem{dfn}{Definition}[section]
    
    \newtheorem{claim}[dfn]{Claim}
    \newtheorem{lem}[dfn]{Lemma}
    \newtheorem{thm}[dfn]{Theorem}

\title{On Approximating the Weighted Region Problem in Square Tessellations\thanks{Supported by JSPS KAKENHI Grant Numbers 23K21646, JP22H05001, and JP21H03397, Japan and JST ERATO Grant Number JPMJER2301, Japan.}}

\author{Naonori Kakimura\thanks{Keio University, Japan. kakimura@math.keio.ac.jp} \and
Rio Katsu\thanks{NS Solutions Corporation, Japan.}
}
\date{}

\begin{document}

\maketitle

\begin{abstract}
The weighted region problem is the problem of  finding the weighted shortest path on a plane consisting of polygonal regions with different weights. 
For the case when the plane is tessellated by squares, we can solve the problem approximately by finding the shortest path on a grid graph defined by placing a vertex at the center of each grid. 
In this note, we show that the obtained path admits $(\sqrt{2}+1)$-approximation.
This improves the previous result of $2\sqrt{2}$.
\end{abstract}

\section{Introduction}

In this note, we study the problem of finding a shortest path in a plane with weighted regions, called the \textit{weighted region problem}.
In the problem, we are given a set of polygonal regions in a plane, where each polygon has a weight.
The goal is to find a path that minimizes the total weight of the path.
The problem was first introduced by Mitchell--Papadimitriou~\cite{MitchellP91}.
They proposed a Dijkstra-like algorithm that finds a $(1+\varepsilon)$-approximate shortest path, running in $O(n^{8}\log{(\frac{n}{\varepsilon})})$ time, where $n$ is the total number of vertices in the polygons.
The running time was later improved to 
$O(\frac{n}{\sqrt{\varepsilon}}\log{\frac{n}{\varepsilon}}\log{\frac{1}{\varepsilon}})$ by Aleksandrov et al.~\cite{AleksandrovMS05}~(see also~\cite{AleksandrovLMS98}).
Recently, it was shown that the weighted region problem is unsolvable in the algebraic computation model over the rational numbers~\cite{CarufelGMOS14,deBerg2024}.

Motivated by the simple, fast computation for the weighted region problem, 
Jaklin et al.~\cite{JaklinTG14} proposed an algorithm that, approximating the weighted region by a simple graph, finds a shortest path on the graph.
Here we focus on the weighted region problem on a square tessellation.
We define a graph consisting of the center points of each square with $8$ neighbors.
See Figure~\ref{fig:example} for an example.
Jaklin et al.~\cite{JaklinTG14} showed that the approximation factor of this approach is at most $4+\sqrt{4-2\sqrt{2}}\approx 5.08$, which was later improved to $2\sqrt{2}\approx 2.82$ by Jaklin~\cite{JaklinPHD}.

    \begin{figure}[t]
\centering
    \begin{tikzpicture}[scale=0.6]
    \draw[thin, ->] (-3.5, 0.5)--(-2.5, 0.5);
    \draw (-2.5,0.5) node[right]{\small$x$};
    \draw[thin, ->] (-3.5, 0.5)--(-3.5, 1.5);
    \draw (-3.5,1.5) node[left]{\small$y$};
    \filldraw[black] (1,0)--(3,0)--(3,2)--(1,2)--cycle;
    \filldraw[gray!30] (1,2)--(4,2)--(4,1)--(5,1)--(5,2)--(6,2)--(6,1)--(7,1)--(7,4)--(6,4)--(6,3)--(3,3)--(3,5)--(1,5)--cycle;
    \filldraw[gray!30] (3,5)--(7,5)--(7,6)--(3,6)--cycle;
    \filldraw[black] (7,0)--(8,0)--(8,1)--(9,1)--(9,5)--(7,5)--cycle;
    \draw[step=1cm,very thin](0,0) grid (10,7);
    \coordinate (s) at (0.5,0.5);
    \coordinate (g) at (9.5,1.5);
    \draw (s) node[below]{$s$};
    \draw (g) node[below]{$g$};
    \draw[red,very thick] (s)--(1,2.5)--(3,3.1)--(7,5)--(9,5)--(g);
    \draw [very thick] (s)--(0.5,1.5)--(1.5,2.5)--(2.5,2.5)--(4.5,4.5)--(6.5,4.5)--(7.5,5.5)--(8.5,5.5)--(9.5,4.5)--(g);
    \fill[black] (s) circle (0.06);
    \fill[black] (g) circle (0.06);
    \draw (4.4,4.4) [above]node{$P^\ast$};
    \draw[red] (4,3.6) [below]node{$\pi^{*}$};
    \end{tikzpicture}
\caption{The white, gray, and black squares have weights $1$, $2$, and $100$, respectively. The red path $\pi^\ast$ depicts a shortest path from $s$ to $g$, and the black path $P^\ast$ is a shortest path on the grid graph.}
\label{fig:example}
\end{figure}
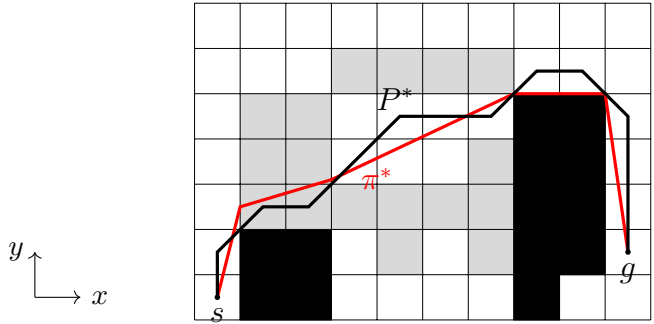  

The main contribution of this note is to improve the ratio to $\sqrt{2}+1\approx 2.41$.
The result is achieved by refining the analysis by Jaklin~\cite{JaklinPHD}.
We remark that there exists an instance such that the shortest path length on the grid graph is at least $\sqrt{2}$ times as long as the original shortest path length.
For the $\{1, +\infty\}$-weight case, the lower bound is known to be $3\sqrt{2}-3$~\cite{BaileyTUKN15}.

Bose et al.~\cite{BoseEuroCG,BoseArxiv} introduced a different kind of a graph where vertices are placed at the corners of each square.
They showed that the shortest path length on the graph admits $\sqrt{4-2\sqrt{2}}$-approximation, which is asymptotically optimal.
See also Bose et al.~\cite{BoseEOS23} for triangular tessellations and Bailey et al.~\cite{BaileyNTK21,BaileyTUKN15} for the $\{1, +\infty\}$-weight case.

\section{Problem Setting}

In this section, we formally describe the problem setting by Jaklin et al.~\cite{JaklinTG14}.
For a positive integer $n$, we denote $[n]=\{1,2,\dots, n\}$.

In a square tessellation, a 2-dimensional plane is divided into $m\times n$ squares with the same size such that each square region has a non-negative weight.
For simplicity, we assume that one side of each square has length $1$.
The whole region consisting of the $m\times n$ squares is denoted by $\mathcal{S}$.
We give an index $(x, y)$ to each square for $x\in [m]$ and $y\in [n]$, where the lower-left corner is $(1,1)$ and the upper \-right corner is $(m,n)$.
The $(x,y)$-th square is denoted by $S_{xy}$.
Let $\alpha_{xy}$ be the weight of the square $S_{xy}$.

In the problem, we are also given a start point $s$ and a goal point $g$, where we assume that $s$ and $g$ are the center points of two distinct squares.
We want to find a path from $s$ to $g$ that minimizes its length.
Here the length of a path is defined as follows.
For a path $\pi$, we denote by $\pi_{xy}$ the segment of $\pi$ included in the $(x, y)$-th square.
Then the length of $\pi$, denoted by $C(\pi)$, is defined as
\[
C(\pi)=\sum_{(x,y)\in [m]\times [n]} \alpha_{xy}\|\pi_{xy}\|,
\]
where $\|\cdot \|$ denotes the Euclidean norm.
Thus the length is the weighted sum of all the path segments.

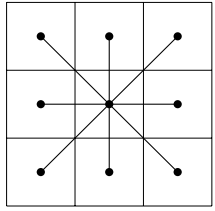
\begin{figure}[t]
\centering
\begin{tikzpicture}[scale=0.9]
    \draw[step=1cm,very thin](0,0) grid (3,3);
    \coordinate (a1) at (0.5,0.5);
    \coordinate (a2) at (1.5,0.5);
    \coordinate (a3) at (2.5,0.5);
    \coordinate (a4) at (0.5,1.5);
    \coordinate (a5) at (1.5,1.5);
    \coordinate (a6) at (2.5,1.5);
    \coordinate (a7) at (0.5,2.5);
    \coordinate (a8) at (1.5,2.5);
    \coordinate (a9) at (2.5,2.5);
    \draw (a5)--(a1);
    \draw (a5)--(a2);
    \draw (a5)--(a3);
    \draw (a5)--(a4);
    \draw (a5)--(a6);
    \draw (a5)--(a7);
    \draw (a5)--(a8);
    \draw (a5)--(a9);
    \fill[black] (a1) circle (0.06);
    \fill[black] (a2) circle (0.06);
    \fill[black] (a3) circle (0.06);
    \fill[black] (a4) circle (0.06);
    \fill[black] (a5) circle (0.06);
    \fill[black] (a6) circle (0.06);
    \fill[black] (a7) circle (0.06);
    \fill[black] (a8) circle (0.06);
    \fill[black] (a9) circle (0.06);
\end{tikzpicture}
\caption{Neighbor vertices of a square.}
\label{fig:8neighborsgridgraph}
\end{figure}

Jaklin et al.~\cite{JaklinTG14} introduced a \textit{grid graph} on a square tessellation, and proposed a method to solve the weighted region problem approximately by finding a shortest path on the grid graph.
Formally, the grid graph is a graph $G=(V,E)$ with vertex set $V$ and edge set $E$ defined as follows.
\begin{itemize}
\item $V=\{v_{xy}\ |\ x\in [m], y\in [n]\}$, and 
\item 
$E= \left\{(v_{xy}, v_{x'y'})
\mid |x-x'|\leq 1 \text{\ and\ } |y-y'|\leq 1\ (x,x'\in [m], y,y'\in [n])\right\}$.
\end{itemize}
Thus each vertex of $G$ has $8$ neighbors~(See Figure~\ref{fig:8neighborsgridgraph}).
We say that an edge $e=(v_{xy}, v_{x'y'})$ is \textit{horizontal} if $y=y'$, \textit{vertical} if $x=x'$, and \textit{diagonal} if $|x-x'|=|y-y'|=1$.

In addition, an edge weight $w: E\to \mathbb{R}_+$ on the grid graph is defined as, for an edge $e=(v_{xy}, v_{x'y'})$, 
\[
w (e) = 
\begin{cases}
\frac{1}{2}(\alpha_{xy}+\alpha_{x'y'}) & \text{if $e$ is horizontal or vertical,}\\
\frac{\sqrt{2}}{2}(\alpha_{xy}+\alpha_{x'y'}) & \text{otherwise}.
\end{cases}
\]
For a path $P$ from $s$ to $g$ on the grid graph $G$, its length is defined to be the sum of the weights of the edges in $P$, denoted by $C(P)$, that is,
$C(P) = \sum_{e\in E(P)}w(e)$.
Also, we denote by $P_{xy}$ the subpath of $P$ included in the $(x,y)$-th square.

\section{Main Analysis}

Let $\pi^\ast$ be a shortest path from $s$ to $g$ on a square tessellation, and $P^\ast$ be a shortest path from $s$ to $g$ on the grid graph $G$.
In this section, we show that the ratio between $C(P^\ast)$ and $C(\pi^{*})$ is at most $\sqrt{2}+1$.

\begin{thm}\label{main3}
For any instance of the weighted region problem on a square tessellation, there exists a  path $P$ from a start point $s$ to a goal point $g$ on the grid graph such that 
\[
C(P)\leq\left(\sqrt{2}+1\right)C(\pi^{*}).
\]
\end{thm}

For a real number $a$, we denote $V_a = \{v_{xy}\mid \|\pi^\ast_{xy}\|\geq a\}$.
Jaklin~\cite{JaklinPHD} showed that the subgraph of $G$ induced by $V_{1/2}$ is connected.
Since $\|P^\ast_{xy}\|\leq \sqrt{2}$ for each square $S_{xy}\in \mathcal{S}$, 
the ratio between $\|P^\ast_{xy}\|$ and $\|\pi^\ast_{xy}\|$ for $v_{xy}\in V_{1/2}$ is at most $2\sqrt{2}$, which implies that $C(P)\leq 2\sqrt{2} C(\pi^{*})$.

To improve the approximation ratio as Theorem~\ref{main3}, we construct a subgraph on $V_a$ for $0\leq a\leq 1/2$, and find a path with a special property~(Lemma~\ref{lem:gogocurry1}).
Then Theorem~\ref{main3} follows by optimizing the parameter $a$.

To this end, we first construct a subgraph $H=(V(H), E(H))$ of the grid graph $G$ as follows.
\[
V(H)=V_a
\ \ \text{and}\ \ 
E(H)=\left\{(v_{xy}, v_{x'y'})\mid |x-x'|+|y-y'|=1
\right\}.
\]
We note that the graph $H$ has only horizontal and vertical edges.
When traversing $\pi^\ast$ from $s$ to $g$, $\pi^\ast$ passes through connected components of $H$.
We denote by  $H_{1},\dots,H_k$ a sequence of connected components of $H$ that $\pi^\ast$ goes through in this order from $s$ to $g$.
Then $s\in V(H_{1})$ and $g\in V(H_k)$.
We remark that the sequence may contain the same connected components more than once.
Let $\mathcal{S}_i=\{S_{xy}\ |\ v_{xy}\in V(H_{i})\}$, and $\overline{\mathcal{S}_i}$ denote the square tessellation covering $\mathcal{S}_i$ for each $i\in [k]$.

We show the following lemma.

\begin{lem}\label{lem:gogocurry2}
Let $a$ be a number with $0\leq a\leq 1/2$.
For any $i\in [k-1]$, 
$\overline{\mathcal{S}_i}$ and $\overline{\mathcal{S}_{i+1}}$ intersect with the corner such that $\pi^\ast$ passes through at least three of the four squares around this corner consecutively.
\end{lem}
\begin{proof}
First observe that $\overline{\mathcal{S}_i}$ and $\overline{\mathcal{S}_{i+1}}$ do not share any edge of squares, as otherwise this
 contradicts that $H_{i}$ and $H_{i+1}$ are different connected components in $H$.

Suppose to the contrary that there exists $i\in [k-1]$ such that $\overline{\mathcal{S}_i}$ and $\overline{\mathcal{S}_{i+1}}$ have no intersection.
Since $\pi^{*}$ is a connected line segment from $s$ to $g$, 
$\pi^{*}$ has a line segment leaving from $\overline{\mathcal{S}_i}$.
We denote by $S_{xy}$ the square that $\pi^{*}$ enters just after $\pi^{*}$ leaves from $\overline{\mathcal{S}_i}$ \textit{for the last time}.
Then, it follows that
$v_{xy}\notin V_a$.
The square that $\pi^{*}$ enters next to $S_{xy}$ is denoted by $S_{x'y'}$.
By the choice of $S_{xy}$, we see that $S_{x'y'}\not\in \mathcal{S}_{i}$.
We note that $|x-x'|\leq 1$ and $|y-y'|\leq 1$ hold.

By symmetry, we may assume that $\pi^\ast$ enters to $S_{xy}$ from $S_{x, y-1}$ and that $S_{x, y-1}\in \mathcal{S}_i$.
Since $v_{xy}\not\in V_a$, that is,  $\|\pi^\ast_{xy}\|<a\leq 1/2$, it holds that $x'=x+1$ or $x'=x-1$.
By symmetry, we may assume that $x'=x+1$.
Moreover, since $v_{xy}\not\in V_a$, it holds that $y'=y$ or $y'=y-1$.

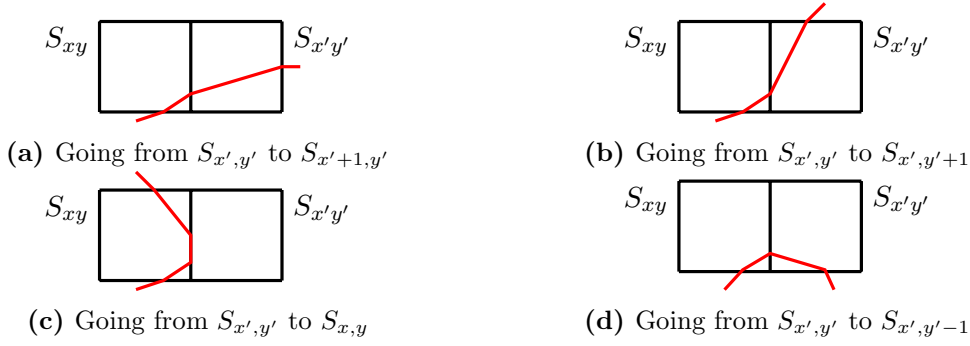
\begin{figure}[t]
    \begin{tabular}{cc}
      \begin{minipage}[t]{0.45\hsize}
       \centering

\begin{tikzpicture}[scale=1.2]
    \draw[step=1cm,very thick](5,5) grid (7,6);
    \draw (5,5.5) [above left]node{\small$S_{xy}$};
    \draw (7,5.5) [above right]node{\small$S_{x'y'}$};
    \draw [red,very thick] (5.4,4.9)--(5.7,5)--(6.0,5.2)--(7,5.5)--(7.2,5.5);
    \end{tikzpicture}
    \subcaption{Going from $S_{x',y'}$ to $S_{x'+1,y'}$}
              \end{minipage} &
      \begin{minipage}[t]{0.45\hsize}
        \centering
\begin{tikzpicture}[scale=1.2]
    \draw[step=1cm,very thick](5,5) grid (7,6);
    \draw (5,5.5) [above left]node{\small$S_{xy}$};
    \draw (7,5.5) [above right]node{\small$S_{x'y'}$};
    \draw [red,very thick] (5.4,4.9)--(5.7,5)--(6.0,5.2)--(6.4,6)--(6.6,6.2);
    \end{tikzpicture}
    \subcaption{Going from $S_{x',y'}$ to $S_{x',y'+1}$}

        \end{minipage} \\
   
      \begin{minipage}[t]{0.45\hsize}
        \centering

    \begin{tikzpicture}[scale=1.2]
    \draw[step=1cm,very thick](5,5) grid (7,6);
    \draw (5,5.5) [above left]node{\small$S_{xy}$};
    \draw (7,5.5) [above right]node{\small$S_{x'y'}$};
    \draw [red,very thick] (5.4,4.9)--(5.7,5)--(6.0,5.2)--(6.0,5.5)--(5.6,6)--(5.4,6.2);
    \end{tikzpicture}
    \subcaption{Going from $S_{x',y'}$ to $S_{x,y}$}

      \end{minipage} &
      \begin{minipage}[t]{0.45\hsize}
        \centering

\begin{tikzpicture}[scale=1.2]
    \draw[step=1cm,very thick](5,5) grid (7,6);
    \draw (5,5.5) [above left]node{\small$S_{xy}$};
    \draw (7,5.5) [above right]node{\small$S_{x'y'}$};
    \draw [red,very thick] (5.5,4.8)--(5.7,5.02)--(6.0,5.2)--(6.6,5.02)--(6.7,4.8);
    \end{tikzpicture}
    \subcaption{Going from $S_{x',y'}$ to $S_{x',y'-1}$}
      
      \end{minipage} 
    \end{tabular}
    \caption{Proof of Lemma~\ref{lem:gogocurry2} when $y'=y$. Only the last case~(d) is possible if $v_{xy}, v_{x'y'}\not\in V_a$.}
    \label{fig:simpleclaim}
  \end{figure}

First suppose that $y'=y$.
If $S_{x'y'}\in \mathcal{S}_{i+1}$, then the lemma holds with the corner between $S_{x,y-1}$ and $S_{x'y'}$.
Thus we may assume that $S_{x'y'}\notin\mathcal{S}_{i+1}$, that is, $v_{x'y'}\notin V_a$.

\begin{claim}\label{clm:simpleclaim}
The shortest path $\pi^\ast$ enters to $S_{x',y'-1}$ after leaving from $S_{x'y'}$.
\end{claim}
\begin{proof}
We observe that $\pi^\ast$ cannot go from $S_{x'y'}$ to $S_{x'+1, y'}$, $S_{x', y'+1}$, or $S_{x, y}$.
In fact, otherwise, we see that $\|\pi^\ast_{xy}\|+\|\pi^\ast_{x'y'}\|\geq 1$, which would contradict that $v_{xy}, v_{x'y'}\not\in V_a$, as $a\leq 1/2$.
See Figure~\ref{fig:simpleclaim} for illustration.
Thus $\pi^\ast$ leaves from $S_{x'y'}$ to  $S_{x',y'-1}$.
\end{proof}

By Claim~\ref{clm:simpleclaim},
$\pi^\ast$ goes from $S_{x'y'}$ to  $S_{x',y'-1}$.
Then, $v_{x',y'-1}$ does not belong to $V_a$, as otherwise $S_{x',y'-1}$ is contained in $\mathcal{S}_i$, which contradicts the choice of $S_{xy}$.
However, applying the same argument as Claim~\ref{clm:simpleclaim} to the two consecutive squares $S_{x'y'}, S_{x',y'-1}$,  we see that $\pi^\ast$ leaves from $S_{x',y'-1}$ to $S_{x'-1, y'-1}$.
See Figure~\ref{fig:dejimon7}.
This contradicts the choice of $S_{xy}$, as $S_{x'-1, y'-1}=S_{x, y-1}\in \mathcal{S}_i$.

\begin{figure}[tbp]
  \begin{minipage}[b]{0.5\columnwidth}
    \centering
    \begin{tikzpicture}[scale=0.8]
    \draw[step=1cm,very thick](5,5) grid (7,6);
    \draw[step=1cm,very thick](6,4) grid (7,5);
    \draw[step=1cm,very thick,blue](4,3) grid (6,5);
    \draw (4,5) [left]node{$\overline{\mathcal{S}_{i}}$};
    \draw (5,6) [left]node{\small$S_{xy}$};
    \draw (7,6) [right]node{\small$S_{x'y'}$};
    \draw (7,4) [right]node{\small$S_{x',y'-1}$};
    \draw[red,very thick] (5,4.8)--(5.6,5)--(6,5.1)--(6.3,5)--(6,4.5);
    \end{tikzpicture}
    \caption{When $\pi^\ast_{xy}$ goes from $S_{x'y'}$ to $S_{x', y'-1}$.}
    \label{fig:dejimon7}
  \end{minipage}
  \begin{minipage}[b]{0.48\columnwidth}
    \centering
\begin{tikzpicture}[scale=1.2]
    \draw[step=1cm,very thick](5,5) grid (6,6);
    \draw[step=1cm,very thick](6,5) grid (7,4);
    \draw[step=1cm,blue, very thick](5,5) grid (6,4);
    \draw (5,4.8) [below left]node{\small$\mathcal{S}_{i}$};
    \draw (5,5.8) [above left]node{\small$S_{xy}$};
    \draw (7,4.8) [below right]node{\small$S_{x'y'}$};
    \draw (7,5.8) [above right]node{\small$S_{x',y'+1}$};
    \draw [red,very thick] (5.2,4.6)--(5.7,5.02)--(6.0,5.02)--(6.6,5.02)--(7,5.7);
    \end{tikzpicture}
    \caption{When $y'=y-1$.}
    \label{fig:cross}

  \end{minipage}
\end{figure}

Next suppose that $y'=y-1$.
Then $\pi^\ast_{xy}$ goes along the bottom side of $S_{xy}$ to the lower-right corner of $S_{xy}$~(or the upper-left corner of $S_{x'y'}$).
See Figure~\ref{fig:cross}.
Since $v_{x'y'}\not\in V_a$, $\pi^\ast_{x'y'}$ goes along the top side of $S_{x'y'}$, as otherwise, $\pi^\ast_{x'y'}$ goes along the left side of $S_{x'y'}$ and enters to $S_{x,y-1}$ in $\mathcal{S}_i$, which contradicts the choice of $S_{xy}$.
Hence $\pi^\ast$ enters to $S_{x', y'+1}$ after leaving from $S_{x'y'}$.
Since $v_{x'y'}\not\in V_a$ and $\|\pi^\ast_{x'y'}\|+\|\pi^\ast_{x',y'+1}\|\geq 1$, we see that $v_{x', y'+1}\in V_a$.
This implies that the lemma holds with the corner between $S_{x, y-1}$ and $S_{x', y'+1}$.

Thus the lemma holds.
\end{proof}

The above lemma implies the existence of a path of the grid graph $G$ such that the used diagonal edges satisfy the following conditions.

\begin{lem}\label{lem:gogocurry1}
For any number $a$ with $0\leq a\leq \frac{1}{2}$,
there exists a path $P$ of the grid graph $G$ from $s$ to $g$ that satisfies the following two conditions.
\begin{itemize}
\item $V(P)\subseteq V_a$.
\item For any diagonal edge $e=(v_{xy},v_{x'y'})$ in $P$,
it holds that $v_{xy'}, v_{x'y}\notin V_a$, and that $\pi^\ast$ passes through at least three of $\{v_{xy}, v_{x'y'}, v_{xy'}, v_{x'y}\}$ in a consecutive way.
\end{itemize}
\end{lem}
\begin{proof}
First consider the case when the graph $H$ defined above is connected.
Then the vertices $s$ and $g$ are contained in the same connected component, which implies that there exists a path from $s$ to $g$ only using horizontal and vertical edges in $H$.
Thus the lemma holds.

Assume that $H$ has at least two connected components.
By Lemma~\ref{lem:gogocurry2}, $\overline{\mathcal{S}_{i}}$ and $\overline{\mathcal{S}_{i+1}}$ intersect with the corner of some square  for any $i\in[k-1]$.
We assume that the corner is of the $(x_i,y_i)$-th square in $\mathcal{S}_i$ and the $(x'_{i+1},y'_{i+1})$-th square in $\mathcal{S}_{i+1}$.
We then append the edge between them to the graph $H$ for each $i\in [k-1]$.
The resulting graph becomes connected, and therefore there exists a path $P$ from $s$ to $g$.
Specifically, the path $P$ is constructed in the following way: (1) For each $i\in [k]$, we find a path $P_i$ of $H_{i}$ from $v_{x'_{i},y'_{i}}$ to $v_{x_{i},y_{i}}$ using only horizontal and vertical edges;
(2) For each $i\in [k-1]$, connect the end vertices of $P_i$ and $P_{i+1}$ with the diagonal edge between $v_{x_{i},y_{i}}$ and $v_{x'_{i+1},y'_{i+1}}$.
The obtained path $P$ satisfies the condition of the lemma, as $V(H)=V_a$.
\end{proof}

For a number $a$ with $0\leq a\leq 1/2$, we denote by $P$ the path satisfying the condition of Lemma~\ref{lem:gogocurry1}.

\begin{lem}\label{dejimon7}
For any $(x, y)$ with $v_{xy}\in V_a$, it holds that
\begin{equation}\label{eq:approx}
\frac{\|P_{xy}\|}{\|\pi^{*}_{xy}\|}
\leq 
\max{\left\{\frac{1}{a},\frac{\sqrt{2}}{1-a}\right\}}.
\end{equation}
\end{lem}
\begin{proof}
By symmetry, the possible patterns of $P_{xy}$ are the $5$ patterns as in Table~\ref{dejimon8}.
We will explain each case separately.
Cases~(i) and~(ii) are obvious, since $\|\pi^\ast_{xy}\|\geq a$ since $v_{xy}\in V_a$.

\begin{table}[t]
  \centering
  \caption{Relationship between $P_{xy}$~(black lines) and $\pi^{*}_{xy}$~(red lines).}
  \label{dejimon8}
  \begin{tabular}{c|c|c|c|c|c} 
 & (i) & (ii) & (iii) & (iv) & (v)\\ \hline
&&&&&\\
 & 
\begin{tikzpicture}[scale=1]
          \draw[step=1cm,very thin](0,0) grid (1,1);
          \draw (0.5,0)--(0.5,1);
          \draw[red] (0.2,0)--(1,0.7);
     \end{tikzpicture}
 &  \begin{tikzpicture}[scale=1]
          \draw[step=1cm,very thin](0,0) grid (1,1);
          \draw (0.5,0)--(0.5,0.5)--(1,0.5);
          \draw[red] (0.2,0)--(1,0.7);
     \end{tikzpicture}
&  \begin{tikzpicture}[scale=1]
          \draw[step=1cm,very thin](0,0) grid (1,1);
          \draw (0,0)--(0.5,0.5)--(1,0.5);
          \draw[red] (0.2,0)--(1,0.7);
     \end{tikzpicture}
& \begin{tikzpicture}[scale=1]
          \draw[step=1cm,very thin](0,0) grid (1,1);
          \draw (0,0)--(1,1);
          \draw[red] (0.2,0)--(1,0.7);
     \end{tikzpicture}
& 
\begin{tikzpicture}[scale=1]
          \draw[step=1cm,very thin](0,0) grid (1,1);
          \draw (0,0)--(0.5,0.5)--(1,0);
          \draw[red] (0,0)--(0.8,0);
     \end{tikzpicture}\\ \hline
 $\|P_{xy}\|$ & $1$ & 1 & $\frac{1+\sqrt{2}}{2}$ & $\sqrt{2}$ & $\sqrt{2}$\\ \hline
 Lower bound of  $\|\pi^{*}_{xy}\|$ & $a$ & $a$ & $1-a$ &  $\min\{1, \sqrt{2}(1-a)\}$ & $1-a$\\ \hline
 Upper bound of $\frac{\|P_{xy}\|}{\|\pi^{*}_{xy}\|}$ & $\frac{1}{a}$ & $\frac{1}{a}$ & $\frac{1+\sqrt{2}}{2(1-a)}$ & $\frac{1}{1-a}$ & $\frac{\sqrt{2}}{1-a}$\\ 
  \end{tabular}
\end{table}

Consider Case~(iii), that is, the $(x, y)$-th square $S_{xy}$ has the subpath $P_{xy}$ in the form of Case~(iii).
Then, $\|P_{xy}\|=(1+\sqrt{2})/2$.
It follows from Lemma~\ref{lem:gogocurry1} that 
$v_{x-1, y}, v_{x, y-1}\not\in V_a$ and 
$v_{x, y}, v_{x-1, y-1}\in V_a$.
See Figure~\ref{dejimon9}.
We will show that, if $\|\pi^\ast_{xy}\|< 1$, then we have $\|\pi^\ast_{xy}\|\geq 1-a$.
Since $\|\pi^\ast_{xy}\|< 1$, $\pi^\ast_{xy}$ is a path between a horizontal side and a vertical side.
We denote by $A$ and $C$ the intersecting 
points of $\pi^\ast_{xy}$ to these sides, respectively, and by $B$ the corner of $S_{xy}$ so that $ABC$ forms a right triangle.
We remark that the point $A$ or $C$ is on the side adjacent to $S_{x-1,y}$ or $S_{x,y-1}$ by Lemma~\ref{lem:gogocurry1}.
We may assume that $A$ is on the bottom side of $S_{x,y}$.
Since $v_{x, y-1}\not\in V_a$, $\|\pi^\ast_{x, y-1}\|$ is at most $a$, and hence we see that the side $AB$ has length at least $1-a$.
Hence, the diagonal side $AC=\|\pi^\ast_{xy}\|$ is of length at least $1-a$. 
Thus, we have $
\frac{\|P_{xy}\|}{\|\pi^{*}_{xy}\|}\leq\frac{1+\sqrt{2}}
{2(1-a)}$.

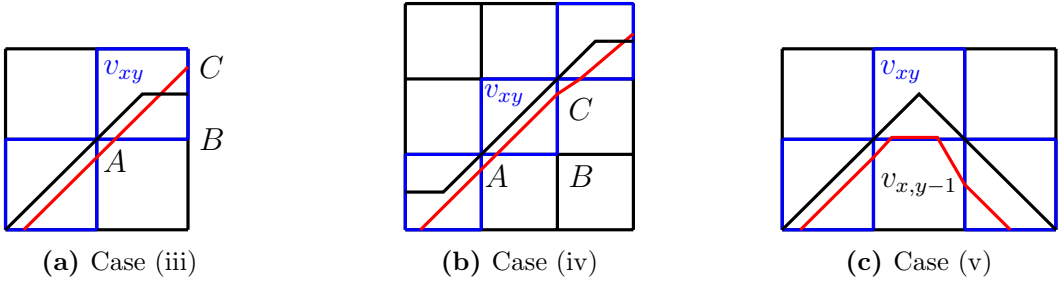
\begin{figure}[t]
\centering
\begin{minipage}[b]{0.32\columnwidth}
    \centering
    \begin{tikzpicture}[scale=1.2]
    \draw[step=1cm,very thick](0,0) grid (2,2);
    \draw[step=1cm,very thick,blue](0,0) grid (1,1);
    \draw[step=1cm,very thick,blue](1,1) grid (2,2);
    \draw[very thick,red] (0.2,0)--(1.2,1)--(2,1.8);
    \draw (1.3,1.5)[blue] node[above]{$v_{xy}$};
    \draw (2,1.8) node[right]{$C$};
    \draw (1,1) ;
    \draw (1.2,1) node[below]{$A$};
    \draw (2,1) node[right]{$B$};
    \draw[very thick] (0,0)--(1.5,1.5)--(2,1.5); 
    \end{tikzpicture}
    \subcaption{Case~(iii)}
    \label{dejimon9}
\end{minipage}
\begin{minipage}[b]{0.32\columnwidth}
    \centering
    \begin{tikzpicture}[scale=1]
    \draw[step=1cm,very thick](0,0) grid (3,3);
    \draw[step=1cm,very thick,blue](0,0) grid (1,1);
    \draw[step=1cm,very thick,blue](1,1) grid (2,2);
    \draw[step=1cm,very thick,blue](2,2) grid (3,3);
    \draw[very thick,red] (0.2,0)--(1.2,1)--(2,1.8)--(2.3,2)--(3,2.6);
    \draw (1.3,1.5)[blue] node[above]{\small $v_{xy}$};
    \draw (2,1.6) node[right]{$C$};
    \draw (1,1) ;
    \draw (1.2,1) node[below]{$A$};
    \draw (2,1) node[below right]{$B$};
    \draw[very thick] (0,0.5)--(0.5,0.5)--(2.5,2.5)--(3,2.5); 
    \end{tikzpicture}
    \subcaption{Case~(iv)}

\end{minipage}
\begin{minipage}[b]{0.32\columnwidth}
    \centering
    \begin{tikzpicture}[scale=1.2]
    \draw[step=1cm,very thick](0,0) grid (3,2);
    \draw[step=1cm,very thick,blue](0,0) grid (1,1);
    \draw[step=1cm,very thick,blue](1,1) grid (2,2);
    \draw[step=1cm,very thick,blue](2,1) grid (3,0);
    \draw[very thick] (0,0)--(1.5,1.5)--(3,0);
    \draw[very thick,red] (0.2,0)--(1,0.8)--(1.2,1.02)--(1.7,1.02)--(2,0.5)--(2.5,0);
    \draw (1.3,1.5)[blue] node[above]{$v_{xy}$};
    \draw (1.5,0.5) node{$v_{x,y-1}$};
     \end{tikzpicture}
    \subcaption{Case~(v)}
\end{minipage}
\caption{Proof of Lemma~\ref{dejimon7}. The blue squares depict $\overline{\mathcal{S}_i}$'s.}
\end{figure}

On Case~(iv), similarly to Case~(iii), we focus on the right triangle with diagonal edge $\|\pi^\ast_{xy}\|$, assuming that $\|\pi^\ast_{xy}\| < 1$.
Then, since  $v_{x-1, y}, v_{x, y-1}\not\in V_a$ and $v_{x+1, y}, v_{x, y+1}\not\in V_a$, we see that the both sides of the triangle are of lengths at least $1-a$.
Hence the diagonal side $\|\pi^\ast_{xy}\|$ has length at least $\sqrt{2}(1-a)$.
Since $\|P_{xy}\|=\sqrt{2}$, we have $
\frac{\|P_{xy}\|}{\|\pi^{*}_{xy}\|}\leq\frac{1}{1-a}$.

On Case~(v), 
assuming that $\|\pi^\ast_{xy}\| < 1$, 
we see by Lemma~\ref{lem:gogocurry1} that $\pi^\ast_{xy}$ is a path such that at least one end point is on the bottom side of $S_{xy}$.
Since $v_{x, y-1}$ is not in $V_a$, we have $\|\pi^\ast_{x, y-1}\|<a$, and hence $\|\pi^\ast_{xy}\|\geq 1-a$.
Thus we have $
\frac{\|P_{xy}\|}{\|\pi^{*}_{xy}\|}\leq\frac{\sqrt{2}}{1-a}$.

Therefore, the ratio can be summarized as in the last row of Table~\ref{dejimon8}, implying that 
\[
\frac{\|P_{xy}\|}{\|\pi^{*}_{xy}\|}\leq\max{\left\{\frac{1}{a},\frac{1+\sqrt{2}}{2(1-a)},\frac{1}{1-a},\frac{\sqrt{2}}{1-a}\right\}}.
\]
Since we have
\[
\frac{1}{1-a}<\frac{1+\sqrt{2}}{2(1-a)}<\frac{\sqrt{2}}{1-a},
\]
the right-hand side can be replaced by $
\max{\left\{\frac{1}{a},\frac{\sqrt{2}}{1-a}\right\}}$.
Thus the lemma holds.
\end{proof}

We are now ready to prove Theorem~\ref{main3}.
The RHS of~\eqref{eq:approx} is minimized by setting $\frac{1}{a}=\frac{\sqrt{2}}{1-a}$.
Hence, when $a=\frac{1}{1+\sqrt{2}}$, 
the ratio between them is maximized, whose value is $\frac{1}{a}=\sqrt{2}+1$.
Thus we have $\|P_{xy}\|\leq (\sqrt{2}+1)\|\pi^{*}_{xy}\|$ for any $(x, y)$ with $v_{xy}\in V_a$.

Therefore, since $P$ is a path of the grid graph $G$, the shortest path $P^\ast$ of $G$ has the length bounded by the following.
\begin{align*}
C(P^\ast)&\leq C(P)=\sum_{(x,y):v_{xy}\in V_a}\alpha_{xy}\|P_{xy}\|
\leq\sum_{(x,y):v_{xy}\in V_a}\alpha_{xy}\left(\sqrt{2}+1\right)\|\pi^{*}_{xy}\|\\
&\leq \left(\sqrt{2}+1\right)\sum_{(x,y)\in[m]\times [n]}\alpha_{xy}\|\pi^{*}_{xy}\|=\left(\sqrt{2}+1\right)C(\pi^{*}).
\end{align*}  
Thus Theorem~\ref{main3} follows.

\bibliography{main}

\end{document}